\documentclass[12pt, a4paper]{article}
\usepackage[english]{babel}
\usepackage[utf8x]{inputenc}
\usepackage{amsmath}
\usepackage{graphicx}
\usepackage{lmodern}
\usepackage{textcomp}
\usepackage[colorinlistoftodos]{todonotes}

\usepackage[mathscr]{euscript}
\usepackage[para,online,flushleft]{threeparttable}
\newcommand\independent{\protect\mathpalette{\protect\independenT}{\perp}}
\def\independenT#1#2{\mathrel{\rlap{$#1#2$}\mkern2mu{#1#2}}}

\usepackage{bm}
\usepackage{amsthm}

\theoremstyle{definition}
\newtheorem{definition}{Definition}[]
\newtheorem{proposition}{Proposition}[]

\usepackage{booktabs,array}

\usepackage[plain,noend]{algorithm2e}

\usepackage{amsmath}
\usepackage{graphicx}

\usepackage{amssymb}

\usepackage[mathscr]{euscript}

\usepackage{titlesec}
\titleformat{\section}[block]
  {\center}
  {\thesection .}
  {1em}
  {\MakeUppercase}
\titleformat{\subsection}[hang]
  {\center \it}
  {\thesubsection}
  {1em}
  {}
\usepackage{authblk}
\author[1]{Anders Huitfeldt}
\author[2,3]{Andrew Goldstein}
\author[4,5]{Sonja A Swanson}
\affil[1]{The Meta-Research Innovation Center at Stanford, Stanford University}
\affil[2]{Department of Medical Informatics, Columbia University}
\affil[3]{Department of Medicine, New York University}
\affil[4]{Department of Epidemiology, Erasmus MC}
\affil[5]{Department of Epidemiology, Harvard T.H. Chan School of Public Health}

\title{The choice of effect measure for binary outcomes: Introducing counterfactual outcome state transition parameters}

\begin{document}
\maketitle

\begin{abstract}

Standard measures of effect, including the risk ratio, the odds ratio, and the risk difference, are associated with a number of well-described shortcomings, and no consensus exists about the conditions under which investigators should choose one effect measure over another. In this paper, we introduce a new framework for reasoning about choice of effect measure by linking two separate versions of the risk ratio to a counterfactual causal model. In our approach, effects are defined in terms of “counterfactual outcome state transition parameters”, that is, the proportion of those individuals who would not have been a case by the end of follow-up if untreated, who would have responded to treatment by becoming a case; and the proportion of those individuals who would have become a case by the end of follow-up if untreated who would have responded to treatment by not becoming a case. Although counterfactual outcome state transition parameters are generally not identified from the data without strong monotonicity assumptions, we show that when they stay constant between populations, there are important implications for model specification, meta-analysis, and research generalization.
\end{abstract}

\section{Background}

Causal effects often differ between groups of people. Consequently, investigators are often required to reason carefully about which measures of effect, if any, can be expected to remain homogeneous between different populations or different subgroups. Investigator beliefs about effect homogeneity have important implications both for model specification, and for the choice of summary metric in meta-analysis \cite{Deeks2008EffectOutcomes}. It is commonly believed that the risk ratio is a more homogeneous effect measure than the risk difference, but recent methodological discussion has questioned the evidence for the conventional wisdom \cite{Poole2015IsMeasure,Panagiotou2015Commentary:Nature.}. 

Approaches to effect homogeneity based on standard effect measures have several noteworthy shortcomings, summarized for reference in Table 1: 

\begin{enumerate}
  \item Approaches that assume equality of the risk ratio or the risk difference may make predictions outside the bounds of valid probabilities.
  \item The odds ratio and the risk ratio have "zero-constraints": If the baseline risk in a population is either 0 or 1, approaches based on assuming that the targeted odds ratio is equal to the odds ratio in a referent population will necessarily result in concluding that the exposure has no effect. \cite{Deeks2002IssuesOutcomes.}
  \item If the risk difference, risk ratio, or odds ratio are equal across different populations, then the proportion that responds to treatment in a given population is required to be a function of that population's baseline risk. \cite{Schechtman2002OddsUse}
  \item The odds ratio is not collapsible; i.e the marginal value of the odds ratio may not be equal to a weighted average of the stratum-specific odds ratios under any weighting scheme, even in the absence of confounding and other forms of structural bias. \cite{Huitfeldt2016OnFramework}
  \item The predictions of the risk ratio model are not symmetric to whether the parameter is based on the probability of having the event, or the probability of not having the event, \textit{i.e} whether we "count the living or the dead". \cite{Sheps1958ShallDead} 
\end{enumerate}

In addition to the five points discussed above, we note that no generally applicable biological mechanism has been 
proposed that would guarantee the risk ratio, the risk difference, or the odds ratio stays constant between different populations. For earlier discussion of how biological mechanisms relate to the choice of effect measure, we point the reader to Siemiatycki and Thomas (1981) and Thompson (1991) \cite{Siemiatycki1981BiologicalCarcinogenesis.} and 
\cite{Thompson1991EffectData.}. 

\begin{table}
\def~{\hphantom{0}}
\small
\caption{Shortcomings of Effect Parameters}{%
\begin{tabular}{p{0.37\linewidth}p{0.12\linewidth}p{0.12\linewidth}p{0.13\linewidth}p{0.13\linewidth}}
%\\
 \\
 \toprule
& \textit{Risk Difference} & \textit{Risk Ratio}&  \textit{Odds Ratio}&  \textit{COST Parameters}\\[5pt]
\midrule
Predicts invalid probabilities & Yes & Yes & No & No \\
Zero-constraints& No & Yes & Yes & No \\
Baseline risk dependence& Yes & Yes & Yes & No \\
Non-collapsibility & No & No & Yes & No \\
Asymmetry to outcome variable & No & Yes & No & No \\
\bottomrule
\end{tabular}}
\label{tablelabel1}
\end{table}

In this paper, we are interested in the effect of binary treatment $A$ (e.g., a drug) on binary outcome $Y$ (e.g., a side effect) in two separate populations $P=s$ and $P=t$. In all examples, we have randomized trial evidence for the average causal effect of the treatment in study population $P=s$ , and wish to predict the effect of introducing the treatment in the target population $P=t$, in which the drug is not available and in which we can only collect observational data. Counterfactuals will be denoted using superscripts. For example, $Y^{a=0}$ is an indicator for whether an individual would have experienced side effect $Y$ if, possibly contrary to fact, she did not initiate treatment with drug \textit{A}. Because we have a randomized trial in study population $s$, the average baseline risk Pr$(Y^{a=0} =1 \vert P=s)$, and the average risk under treatment, Pr$(Y^{a=1} =1 \vert P=s)$ are identified from the data. Since treatment is not available in target population $t$, the average baseline risk, Pr$(Y^{a=0} =1 \vert P=t)$, is also identified from the data. Our goal is to use this information, in combination with some plausible assumption about effect homogeneity, to predict the average risk under treatment in the target population, Pr$(Y^{a=1} =1 \vert P=t)$.\footnote{We note that considerations other than effect heterogeneity may also be relevant to the choice of effect measure. In particular, while decision-making would benefit from information on both Pr$(Y^{a=1}\vert P=t)$ and Pr$(Y^{a=0} \vert P=t)$ in order to weigh the costs and benefits of the intervention across different outcomes, it may be sufficient to know the risk difference (RD), Pr$(Y^{a=1} - Y^{a=0}\vert  P=t)$, for example if the decision-maker is neither risk-averse nor risk-seeking over the relevant outcome (that is, if their preferences can be represented by a social welfare function which increases linearly with the number of averted incident cases). This has been used as an argument in favor of measuring effects on the additive scale. Others have argued that one should give preference to summary statistics which are readily understood by the public, and that inverse measures of effect such as the number needed to treat (NNT) are better suited for this purpose. The focus of the current paper is on situations where the investigators must present a summary measure of an effect in a study population, possibly conditional on a set of effect modifiers, for potential use across a range of different target populations. In such situations, generalizability is prioritized over these other considerations, at least concerning the choice of initial effect measure.}

Briefly, we recall that VanderWeele (2012) \cite{VanderWeele2012ConfoundingMeasure.} defined two separate types of effect heterogeneity: Effect modification in distribution (where the distributions of counterfactual variables vary across  populations) and effect modification in measure (where a particular effect measure varies across populations). Here, we will propose a novel approach to effect homogeneity based on a new class of non-identified effect parameters; and show how this framework can sometimes be used to reason about homogeneity in terms of standard, identifiable measures of effect such as the risk difference and the risk ratio. Specifically, we will consider two risk ratios, which are equivalent to the two risk ratios considered by Deeks (2012) \cite{Deeks2002IssuesOutcomes.}, and defined as follows:

\begin{center}

\[RR(-) = \frac{\text{Pr}(Y^{a=1} =1)}{\text{Pr}(Y^{a=0} =1)}\]  
\[RR(+) = \frac{1- \text{Pr}(Y^{a=1} =1)}{1-\text{Pr}(Y^{a=0} =1)} = \frac{\text{Pr}(Y^{a=1} =0)}{\text{Pr}(Y^{a=0} =0)} \] 

\end{center}

This paper is organized as follows. In part 2, we describe a scenario to illustrate the motivation for expanding upon existing effect parameters. In part 3, we introduce counterfactual outcome state transition (COST) parameters, and propose a definition of effect homogeneity based on these parameters. In part 4, we discuss the conditions under which COST parameters are identified from the data, and the implications of violations of these conditions. In part 5, we show that the COST parameters are not symmetric to the coding of the exposure variable, and discuss the implications of this observation. In part 6, we link COST parameters to substantive knowledge by providing examples of biological processes that result in effect homogeneity. In part 7, we discuss the empirical implications of our model. We conclude in part 8. 

\section{Motivating Example}

Suppose a team of investigators have data from a randomized trial where the risk of a particular side effect was 2\% under no treatment and 3\% under treatment. They wish to predict the effect in a different target population, where the baseline risk of the side effect is 10\% (Table 2). For simplicity, we postulate that this treatment does not prevent the outcome from occurring in any individual (``monotonic effect''); this assumption will be discussed in detail later in the paper. 

While discussing their data analysis plan, the first investigator postulates that the standard risk ratio $RR(-)$ is equal between these two populations. Under this assumption, he estimates that 15\% of the target population will experience the side effect under treatment. However, a second investigator believes that $RR(+)$, not $RR(-)$, might be equal between the two populations. Under this assumption, he estimates that 10.9\% of the target population will experience the side effect under treatment.

A third investigator notes that neither the first nor second investigator have substantive arguments for choosing the $RR(+)$ versus the $RR(-)$. However, she realizes that the group at risk of being adversely affected by treatment is the 90\% of the target population who were originally destined \textit{not} to experience the side effect. She points out that the first investigator's assumption (i.e. that $RR(-)$ was equal between the populations) results in a prediction that among people who are at risk of being adversely affected by the treatment, the proportion who respond is much higher in population $t$ than in population $s$, simply because the baseline risk is higher. 

Given this, the third investigator instead suggests assuming a specific probability is constant across the populations: the probability that a person who was previously not destined to experience the outcome, does not experience the outcome in response to treatment. This probability can be computed in the trial as $\frac{97\%}{98\%}\approx 99\%$ and applied to the 90\% who were originally destined to be unaffected is the target population. In this case, the third investigator's approach results in the exact same estimate as the second investigator's approach: an estimated 10.9\% of the target population will experience the side effect under treatment. 

\begin{table}

\small
\def~{\hphantom{0}}
\caption{Predicted results using different effect measures}{%
\begin{tabular}{p{0.32\linewidth}p{0.11\linewidth}p{0.11\linewidth}p{0.11\linewidth}p{0.11\linewidth}}
%\\
 \\
 \toprule
&  $RR(-)$ & $RR(+)$& $RD$& $OR$ \\[5pt]
\midrule
Pr$(Y^{a=0} =1 \vert P=s)$ & 2\% & 2\% & 2\% & 2\% \\ 
Pr$(Y^{a=1} =1 \vert P=s)$ & 3\% & 3\% & 3\% & 3\% \\ 
Effect & $RR(-)$=1.5 & $RR(+)$=0.99 & $RD$=0.01 & $OR$=1.515 \\
Pr$(Y^{a=0} =1 \vert P=t)$ & 10\% & 10\% & 10\% & 10\% \\ 
Predicted Pr$(Y^{a=1} =1 \vert P=t)$& 15\% & 10.9\% & 11\% & 14.4\%\\
\bottomrule
\end{tabular}}
\label{tablelabel2}
\end{table}

Variations of the third investigator's arguments have arisen independently multiple times in the literature, dating back more than half a century \cite{Sheps1958ShallDead,Deeks2002IssuesOutcomes.}, but these recommendations have rarely been translated into common practice. Throughout the rest of this paper, we will formalize this line of reasoning in a counterfactual causal model, in order to explore its scope, limits and implications.

\section{Definition of the Counterfactual Outcome State Transition Parameters}

We define ``counterfactual outcome state transition'' (COST) parameters based on the probability that a person who becomes a case if untreated remains a case if treated, and by the probability that a person who does not become a case if untreated remains a non-case if treated. We also offer interpretations of these quantities in terms of the four deterministic response types \cite{Greenland1986IdentifiabilityConfounding.} (Table 3). A list of parameters considered in this paper is shown in Table 4.

\begin{definition}
$G$ is defined as the probability of being a case if treated, among those who would have been a case if untreated:  \[G=\text{Pr}(Y^{a=1}=1\vert Y^{a=0}=1)\] 
\end{definition}

\begin{definition}
$H$ is defined as the probability of not being a case if treated, among those who would not have been a case if untreated:  \[\ H=\text{Pr}(Y^{a=1}=0\vert Y^{a=0}=0)\] 
\end{definition}

In a deterministic model, $G$ can be interpreted as the proportion who are ``Doomed'', among those who are either ``Doomed'' or ``Preventative'' and $H$ can be interpreted the proportion who are ``Immune'', among those who are either ``Immune'' or ``Causal''. We will refer to $G$ and $H$ as the COST parameters for introducing treatment. $G$ and $H$  can be indexed for a specific population with subscripts (e.g., $G_s$ is the parameter $G$ in population $s$). If the COST parameters for introducing treatment are equal between two populations (i.e. if $G_t= G_s$ and $H_t= H_s$,) this can equivalently be written as \[Y^{a=1}\independent P \vert Y^{a=0}\]
Similar conditions were considered in different contexts by Gechter (2015) \cite{Gechter2015GeneralizingIndia} and by Athey and Imbens (2006) \cite{Athey2006IdentificationModels}. In our setting of binary outcomes, this definition of effect homogeneity addresses all previously discussed limitations of the standard effect measures: The underlying parameters have no baseline risk dependence, do not produce logically impossible results, are collapsible over arbitrary baseline covariates, and have no zero constraints. Further, if the coding of the outcome is altered, the only consequence is that the values of $G$ and $H$ are reversed. (See appendix 1)

\begin{table}
\def~{\hphantom{0}}
\small
\begin{center}
\begin{threeparttable}
\caption{Counterfactual Response Types for Binary Outcomes}{%
\begin{tabular}
{p{0.15\linewidth}p{0.35\linewidth}p{0.25\linewidth}}
%\\
 \\
 \toprule
\textit{Type of individual} & \textit{Description of response type} & \textit{Potential outcomes} \\[5pt]
\midrule
Type 1& No effect of treatment
(Individual \textit{Doomed} to get the disease with respect to exposure)
 & $Y^{a=0}=1$, $Y^{a=1}=1$ \\ 
Type 2 & Exposure \textit{Causative} (Individual susceptible to exposure)
 & $Y^{a=0}=0$, $Y^{a=1}=1$\\ 
Type 3 & Exposure \textit{Preventative}
(Individual susceptible to exposure) & $Y^{a=0}=1$, $Y^{a=1}=0$\\
Type 4 & No effect of treatment (Individual \textit{Immune} from disease with respect to exposure) & $Y^{a=0}=0$, $Y^{a=1}=0$\\
\bottomrule
\end{tabular}}
\label{tablelabel3}
\begin{tablenotes}
Note that, under these definitions, it readily follows that Pr$(Y^{a=0}=1)$ = Pr(Doomed) + Pr(Preventative) and that Pr$(Y^{a=1}=1)$= Pr(Doomed) + Pr(Causal).
\end{tablenotes}
\end{threeparttable}
\end{center}
\end{table}

\begin{table}
\begin{center}
\begin{threeparttable}
\caption{List of parameters}{%
\small
\begin{tabular}{p{0.09\linewidth}p{0.30\linewidth}p{0.20\linewidth}p{0.25\linewidth}}
%\\
 \\
 \toprule
\textit{Parameter} & \textit{Definition}& \textit{Key conditions for identification}& \textit{Identifying expression} \\[5pt]
\midrule
$RR(-)$ & \(\frac{\text{Pr}(Y^{a=1}=1)}{\text{Pr}(Y^{a=0}=1)}\)  & No confounding & \( \frac{\text{Pr}(Y=1 \vert A=1 )}{\text{Pr}(Y=1\vert A=0)}\) \\ 
$RR(+)$ & \( \frac{\text{Pr}(Y^{a=1}=0)}{\text{Pr}(Y^{a=0}=0)} \) & No confounding & \( \frac{\text{Pr}(Y=0 \vert A=1 )}{\text{Pr}(Y=0 \vert A=0)}\)  \\ 
$G$ & \(\text{Pr}(Y^{a=1}=1\vert Y^{a=0}=1)\)  & No confounding and non-increasing monotonicity & $RR(-)$ \\ 
$H$ & \(\text{Pr}(Y^{a=1}=0\vert Y^{a=0}=0)\)  & No confounding and non-decreasing monotonicity & $RR(+)$ \\ 
$I$ & \(\text{Pr}(Y^{a=0}=1\vert Y^{a=1}=1)\)  & No confounding and non-decreasing monotonicity & \( \frac{1}{RR(-)} \)\\ 
$J$ & \(\text{Pr}(Y^{a=0}=0\vert Y^{a=1}=0)\)  & No confounding and non-increasing monotonicity & \( \frac{1}{RR(+)} \)\\ 
\bottomrule
\end{tabular}}
\label{tablelabel4}
\begin{tablenotes}
All parameters shown above are defined separately in populations $s$ and $t$. Whenever we need to clarify the population in which the parameter is defined, subscripts are used (e.g., $RR(-)_s$).
\end{tablenotes}
\end{threeparttable}
\end{center}
\end{table}

\section{Identification of the Counterfactual Outcome State Transition Parameters}

Now that we have defined COST parameters, and described their motivation and attractiveness, we proceed to discuss how we can compute them in our studies. In an ideal randomized trial, we can identify the standard effect measures (RR, RD) without further assumptions beyond those expected to hold by design. Unfortunately, this is not the case for COST parameters: COST parameters are generally not identifiable without further assumptions. Therefore, even if we (somehow) knew that the parameters \textit{G} and \textit{H} are equal between two populations, we may not be able to use this fact to predict what happens when we introduce treatment in the target population. We next introduce assumptions that lead to identifiability (Propositions 1, 4), and discuss scientific and clinical implications when we are (Propositions 2, 5) and when we are not (Propositions 3, 6) willing to make those assumptions. The key identifiability assumption is monotonicity \cite{VanderWeele2010SignedInference., VanderWeele2009PropertiesGraphs}: We say there is non-increasing monotonicity if individuals who do not get the outcome if untreated, do not get the outcome if treated: Pr$(Y^{a=1}=1 \vert Y^{a=0}=0) = 0$. In other words, $H=1$. Similarly, non-decreasing monotonicity occurs if individuals who get the outcome if untreated also get the outcome if treated, in which case $G=1$. 

Monotonicity is defined in the context of the specific exposure-outcome under consideration, and its plausibility therefore needs to be evaluated on a case-by-case basis. This means that within a trial, we may need to consider the plausibility of monotonocity for each outcome of interest separately. For instance, the antiarrhythmic agent Amiodarone can cause arrhythmia in some individuals and prevent it in others, and therefore does not have monotonic effects for the arrhythmia outcome. However, if the outcome of interest is a side effect such as pulmonary fibrosis, monotonicity may be a viable assumption because the drug is unlikely to prevent any individual from getting pulmonary fibrosis. In general, monotonicity could be a more reasonable approximation in situations where the outcome is a side effect strongly associated with pharmacological treatment.

\begin{proposition}
\label{proposition1}
If treatment monotonically reduces the incidence of event \textit{Y}, then \textit{G} is identified from the data of a randomized trial, and is equal to the standard risk ratio, \textit{RR(-).}
\end{proposition}

\begin{proposition}
\label{proposition2}
If the counterfactual outcome state transition parameters for introducing treatment are equal between populations $s$ and $t$, and if treatment monotonically reduces the risk of the outcome, then $RR(-)$ in the target population is equal to $RR(-)$ in the study population.
\end{proposition}

\begin{proposition}
\label{proposition3}
If the counterfactual outcome state transition parameters for introducing treatment are equal between populations $s$ and $t$, and if treatment reduces the incidence of $Y$ but not monotonically, then $RR(-)$ in the study population is a biased estimate of $RR(-)$ in the target population.\footnote{If there is substantial non-monotonicity or if the baseline risks differ substantially between the populations, the bias term gets very large, resulting in highly biased estimates which may even be on the wrong side of 1. For example, if $G$ is 0.05 and $H$ is 0.99 in both populations, and the baseline risk is 0.005 in the study population but 0.05 in the target population, then the true risk ratio in the target population is 0.24, but the risk ratio in the study population is 2.05. We therefore caution against using $RR(-)_s$ as an approximation of  $RR(-)_t$ unless it is plausible, in both populations, to expect "near-monotonicity", defined as $G \times $Pr$(Y^{a=0}=1) \gg (1-H) \times ($Pr$(Y^{a=0}=0)$. This condition, which ensures that $G$ is approximately equal to $RR(-)$, will be met if the prevalence of the "doomed" response type is much higher than the prevalence of the "causative" response type. In the setting of a randomized trial, this might be a reasonable approximation if there is reason to believe that a great majority of the events in the treatment arm would have occurred even in the absence of treatment.} If the baseline risk in the target population is lower than in the study population, then the true risk under treatment in the target population is higher than what would be predicted by assuming $RR(-)_s = RR(-)_t$, whereas the opposite holds if the baseline risk is higher in the target population. The magnitude of the bias is a function of the extent of non-monotonicity, and of the ratio of baseline risk between the two populations (see appendix 2).
\end{proposition}

The following three propositions are exactly symmetric to the preceding results, but apply to situations where exposure increases the risk of the outcome (i.e., monotonicity holds in the other direction). In such situations, we identify 
$H$ instead of $G$:

\begin{proposition}
\label{proposition4}
If treatment monotonically increases the incidence of event $Y$, then $H$ is identified from the data of a randomized trial, and is equal to the recoded risk ratio $RR(+)$
\end{proposition}

\begin{proposition}
\label{proposition5}
If the counterfactual outcome state transition parameters for introducing treatment are equal between populations $s$ and $t$, and if treatment monotonically increases the risk of the outcome, then $RR(+)$ in the study population is equal to $RR(+)$ in the target population.
\end{proposition}

\begin{proposition}
\label{proposition6}
If the counterfactual outcome state transition parameters for introducing treatment are equal between populations $s$ and $t$, and if treatment increases the incidence of $Y$ but not monotonically, then $RR(+)$ in the study population is a biased estimate of $RR(+)$ in the target population.
\end{proposition}

Proofs of propositions (1-6) are provided in Appendix 2. 

\section{Asymmetry to the Coding of the Exposure Variable}

So far, we have focused on problems -- and resolutions to some problems -- related to how the investigator encodes the \textit{outcome} variable in the database. However, COST parameters are not invariant to the coding of the exposure variable. To illustrate, reconsider the trial where the risk under treatment was 3\% and the risk under no treatment was 2\%. If we reverse the coding of the exposure variable, we notice that the new exposure variable in fact \textit{reduces} the risk of the outcome, meaning that a naive application of our approach would suggest using the standard (but now inverted) risk ratio $\frac{1}{RR(-)}$ to estimate the probability of being unaffected by treatment. However, there is a subtle but important distinction from the earlier approach: by changing  the definition of exposure, we have also changed the meaning of the parameter so that it is now defined as the probability of being unaffected by treatment among those who would have become cases \textit{under treatment}, rather than among those who would have become cases under no treatment. This can be conceptualized as the effect of \textit{removing} treatment from a fully treated population.

We will refer to the COST parameters associated with removing treatment as $I$ and $J$. The choice of coding of the exposure variable is then equivalent to choosing whether to model equality of effect based on the parameters $G$ and $H$, or based on the parameters $I$ and $J$. For notational simplicity, we will continue to use the original coding of the exposure variable, and instead frame the question in terms of whether an investigator should define equality of effects based on the parameters $G$ and $H$, or the parameters 
$I$ and $J$.

\begin{definition}
$I$ is defined as the probability of being a case if untreated, among those who would have been a case if treated:  \[I=\text{Pr}(Y^{a=0}=1\vert Y^{a=1}=1)\] 
\end{definition}

\begin{definition}
$J$ is defined as the probability of not being a case if untreated, among those who would not have been a case if treated:  \[\ J=\text{Pr}(Y^{a=0}=0\vert Y^{a=1}=0)\] 
\end{definition}

In a deterministic model, $I$ can be interpreted as the proportion who are ``Doomed'', among those who are either ``Doomed'' or ``Causal'', and $J$ can be interpreted the proportion who are ``Immune'', among those who are either ``Immune'' or ``Preventative''. If the COST parameters for removing treatment are equal between two populations (i.e. if $I_t= I_s$ and $J_t= J_s$,) this can equivalently be written as \(Y^{a=0}\independent P \vert Y^{a=1}\). With these definitions, results exactly analogous to propositions 1 through 6 can be derived for $I$ and $J$.

As implied earlier in this section, one can construct examples to show that equality of COST parameters for introducing treatment does not imply equality of COST parameters removing treatment. In fact, if the baseline risks differ, then the two homogeneity conditions rarely hold simultaneously (an obvious exception to this would be under the sharp causal null hypothesis). This observation arguably presents a major conceptual challenge to the claim that our definition captures the intuitive idea of ``equal effects'': to make our inferences invariant to the coding of the outcome, we have made them dependent on the coding of the exposure. However, in the next section, we show that, with some background knowledge about biological mechanisms, it is possible to reason about whether the COST parameters for introducing treatment are more likely to be homogeneous than the COST parameters for removing treatment. In particular, we show that it is possible to provide models for the data-generating process that guarantee equality of the COST parameters for introducing treatment, equality of the COST parameters for removing treatment, or neither.

\section{Biological Knowledge and Equality of Treatment Effects}

While it is usually not possible to reason \textit{a priori} about whether the risk difference or odds ratio are equal between populations, we now provide a simple example to show that it is sometimes possible to reason based on biological knowledge about whether the COST parameters are equal between populations. This is intended only as a proof-of-concept, and the example is purposefully oversimplified in order to illustrate the principles. An outline of a formal treatment of these ideas is provided in Appendix 3.  

Consider a team of investigators who are interested in the effect of antibiotic treatment on mortality in patients with a specific bacterial infection. Since this antibiotic is known to reduce mortality, the investigators need to decide whether to report $RR(-)$ as an approximation of $G$, or alternatively $RR(+)$ as an approximation of $\frac{1}{J}$. In order to ensure external validity, this choice will be determined by their beliefs about which parameter is most likely to be constant across populations.  

The investigators believe that the response to this antibiotic is completely determined by an unmeasured bacterial gene, such that only those who are infected with a bacterial strain with this gene respond to treatment. The prevalence of this bacterial gene is equal between populations, because the populations share the same bacterial ecosystem. If, as seems likely, the investigators further believe that the gene for susceptibility reduces mortality in the presence of antibiotics, but has no effect in the absence of antibiotics, they will conclude that $G$ may be equal between populations. If, on the other hand, they had concluded that the gene for susceptibility causes mortality in the absence of antibiotics but has no effect in the presence of antibiotics, they would instead expect equality of $J$ across populations. 

For many antimicrobial therapies, the microbial genes that determine antibiotic susceptibility generally have functions that are perhaps better approximated by the first approach, and therefore motivates the choice to model the data as if $G$ is equal between the populations. In other situations, in the presence of different subject matter knowledge, similar logic could be used to reach different conclusions. One example of this occurs in pharmacological applications involving adverse reactions to drugs, where it may be reasonable to use the parameter $H$ if the determinants of susceptibility are equally distributed between populations, under certain assumptions about how those determinants interact with the drug. In many realistic applications, it may be more plausible that the COST parameters for introducing treatment are equal than the corresponding parameters for removing treatment, but this is by no means universal: for example, if some humans had retained our ancestors' ability to synthesize Vitamin C endogenously, then the effect of fresh fruit on scurvy might be better modeled based on the parameter $J$.

In most settings, the particular function of the attribute that determines treatment susceptibility will not be known. In such situations, it is necessary to reason on theoretical grounds about which model is a better approximation of reality. One possible approach to determine whether either type of effect equality is biologically plausible would be to consider how an attribute or gene with the necessary function avoided either reaching fixation or being selected out of existence. For example, a genotype that causes the outcome in the presence of exposure will very likely go extinct in a population where everyone is exposed, but may survive in a proportion of a population where everyone is generally unexposed. This equilibrium may be equal between different groups of people (for example, equal between men and women in the same gene pool). Therefore, if treatment was unavailable in recent evolutionary history, the COST parameters for introducing introducing treatment may be more likely to be equal than the COST parameters for removing treatment. Similarly, an attribute that prevents the outcome in the absence of exposure will quickly reach fixation if everyone is unexposed, but its absence may survive in a small, stable fraction of the population if everyone is exposed. Therefore, if everyone were exposed in recent evolutionary history, the COST parameters for removing treatment may be more likely to be equal than the COST parameters for introducing treatment. Thus, this line of reasoning may provide an additional viable argument for choosing the index level of the exposure variable based on the value which it took by default in recent evolutionary history.

\section{Testing for Heterogeneity}

If we believe that the COST parameters for introducing treatment are equal across populations, an empirical implication is that meta-analysis based on $RR(-)$ will be less heterogeneous for exposures which monotonically reduce the incidence of the outcome, whereas meta-analysis based on $RR(+)$ will be less heterogeneous for exposures which monotonically increase the incidence of the outcome. 

However, this observation is complicated by an additional asymmetry associated with the ratio scale: If the outcome is rare (which is usually the case), Pr$(Y^{a=1}=0)$ and Pr$(Y^{a=0}=0)$ will both be close to $1$, and $RR(+)$ will therefore also be close to $1$, even if treatment has a substantial effect. This results in a compression of the $RR(+)$ scale, such that when heterogeneity is measured in terms of the absolute differences between effect sizes, clinically meaningful heterogeneity between populations will only be apparent at the second or even third decimal space. In contrast, heterogeneity on the $RR(-)$ scale generally manifests itself at the first decimal. Therefore, any attempt to measure heterogeneity based on the absolute difference between each study's estimate and the overall meta-analytic estimate will result in higher values for $RR(-)$ than $RR(+)$ for rare outcomes, for reasons that arguably say more about the mathematical differences between the scales than about their relative usefulness for summarizing effect sizes.

One option may be to quantify each study's deviation from the common effect in terms of the lowest possible proportion of individuals whose outcome variable must be ``switched" in order for the effect estimate in the study to equal the overall meta-analytic estimate. Another potential approach to this problem is to use the Risk Difference ($RD$) in place of $RR(+)$ for exposures which increase the incidence of $Y$. In Appendix 4, we show that if the outcome is rare, effect homogeneity on the $RR(+)$ scale implies near-homogeneity on the risk difference scale. This approach is consistent with previous suggestions to consider the ``relative benefits and absolute harms" \cite{Glasziou1995AnTreatment} of medical interventions. Investigators using this approach must further keep in mind the potential differences in power between tests for homogeneity on the additive and multiplicative scales. \cite{Poole2015IsMeasure}

\section{Conclusion}

We have proposed a new approach to considering effect equality across populations, which avoids several well-established shortcomings of definitions based on standard effect measures. Our approach distinguishes equality of the effect of introducing the treatment to a fully untreated population from equality of the effect of removing the treatment from a fully treated population; and therefore requires investigators to reason carefully about the distinction between the two homogeneity conditions. Further, we provided examples of biological models that correspond to each form of effect equality. While the utility of our approach is limited to the restricted range of applications where these biological models are a reasonable approximation of reality, we believe such applications could occur with some frequency when studying the effectiveness and safety of pharmaceuticals. 

If investigators are willing to assume that the COST parameters for introducing treatment are equal between populations, it follows that the standard risk ratio $RR(-)$ should be used for exposures which monotonically reduce the incidence of the outcome, and that the recoded risk ratio $RR(+)$  should be used for exposures which monotonically increase the risk of the outcome. Thus, the risk ratio will generally be constrained between 0 and 1. If the outcome is rare, $RD$ may be used in the place of $RR(+)$ for exposures that increase the incidence of the outcome. If the effect of exposure is not monotonic, the investigator may still choose the risk ratio model based on whether treatment increases or decreases the risk on average; in such situations, the extent of bias will be small if the extent of non-monotonicity is small, or if the populations have comparable baseline risks. This approximation is highly sensitive to violations of these conditions, and if there is reason to suspect substantial non-monotonicity, identification is not feasible and investigators may consider using an alternative approach to effect homogeneity [\cite{Cole2010GeneralizingTrial.,Bareinboim2013AResults}].

\section*{Appendix 1: Properties of COST parameters}

{
\setlength{\parindent}{0pt}

\textit{Valid predictions:} The COST parameter approach results in predictions for Pr$(Y^{a=1} =1 \vert P=t)$  that are valid probabilities, i.e. that are contained in the interval [0,1]. The predictions are generally of the form $\hat{\text{Pr}}(Y^{a=1}=1 \vert P=t) = {\text{Pr}}(Y^{a=0}=1 \vert P=t) \times G + (1-\text{Pr}(Y^{a=0}=1 \vert P=t)  )\times (1-H)$.  In other words, $\hat{\text{Pr}}(Y^{a=1}=1 \vert P=t)$ is a weighted average of $G$ and 1-$H$. Since both $G$ and 1-$H$ are contained in [0,1], the prediction is also in this interval.\newline

\textit{No Zero-Constraints:} Here, we show that for any baseline risk in the target population Pr$(Y^{a=0}=1 \vert P=t)$, there exist possible values of the COST parameters such that Pr$(Y^{a=1} =1 \vert P=t) \neq $Pr$(Y^{a=0}=1 \vert P=t)$.  For any baseline risk other than 0, it can easily be seen that such values exist, for example if $G \neq 0$ and $H=0$. If the baseline risk is zero, such values of the parameters also exist, for example if $G=0$ and $H \neq 0$ \newline

\textit{Baseline risk dependence:} Here we define a measure of effect to be baseline risk dependent if, in order for the parameter to stay equal between populations, it is necessary that the proportion of individuals who respond to treatment (by experiencing the opposite outcome of what they would have experienced in the absence of treatment) varies with baseline risk. One can easily observe that COST parameters are not affected by such baseline risk dependence: This follows almost by definition, since COST parameters were designed specifically to avoid this form of baseline risk dependence. \newline

\textit{Collapsibility:} $G$ is collapsible if, for any baseline covariates $V$, there exist weights $W_v$ such that $G = \frac{\sum{G_v * W_v}}{\sum{W_v}}$. The weights $W_v = $Pr$(V=v \vert Y^{a=0}= 1)$ always satisfy this equation; the proof of this is exactly analogous to the corresponding proof for the risk ratio [\cite{Huitfeldt2016OnFramework}]. Analogously, the weights for $H$ are Pr$(V=v \vert Y^{a=0}= 0)$, the weights for  $I$ are Pr$(V=v \vert Y^{a=1}= 1)$ and the weights for $J$ are Pr$(V=v \vert Y^{a=1}= 0)$. \newline

\textit{Symmetry:} If the coding of the outcome variable is reversed, the only consequence is that $H$ and $G$ change value. To illustrate this, we will discuss variables and parameters that are defined according to the recoded outcome; these are denoted with a star (i.e. $Y_* = (1-Y)$). Recall that we defined $G=\text{Pr}(Y^{a=1}=1\vert Y^{a=0}=1)$. If we reverse the coding, $G_* =\text{Pr}(Y_*^{a=1}=1\vert Y_*^{a=0}=1)$. By replacing $Y_*$ with $Y$, this can be written as $G_*=\text{Pr}(Y^{a=1}=0\vert Y^{a=0}=0)$, which is equal to $H$. The same logic can be used to show that $H_* =G$.\newline

\section*{Appendix 2: Proofs of propositions 1-6}

In the following proofs, we will simplify the notation by defining \(S_0 = \text{Pr}(Y^{a=0} =1 \vert P=s)\), \(S_1 = \text{Pr}(Y^{a=1} =1 \vert P=s)\), \(T_0 = \text{Pr}(Y^{a=0} =1 \vert P=t)\) and \(T_1 = \text{Pr}(Y^{a=1} =1 \vert P=t)\).

\begin{proof}[Proof of Proposition~\ref{proposition1}]

\bigskip Note, first, that the risk under no treatment and the risk under treatment can be rewritten in terms of response types:
\[S_0 = \text{Pr}(Y^{a=0}=1,Y^{a=1}=1)  + \text{Pr}(Y^{a=0}=1,Y^{a=1}=0)\]
\[S_1 = \text{Pr}(Y^{a=0}=1,Y^{a=1}=1)  + \text{Pr}(Y^{a=0}=0,Y^{a=1}=1)\]

Moreover, because the response types are mutually exclusive and collectively exhaustive events, it follows that:

\[1- S_0 = \text{Pr}(Y^{a=0}=0,Y^{a=1}=0) + \text{Pr}(Y^{a=0}=0,Y^{a=1}=1)\]
\[1-S_1 = \text{Pr}(Y^{a=0}=0,Y^{a=1}=0) + \text{Pr}(Y^{a=0}=1,Y^{a=1}=0)\]

Recall that $G$ and $H$ can equivalently be defined as follows:
 
 \[G = \frac{\text{Pr}(Y^{a=0}=1, Y^{a=1}=1)}{\text{Pr}(Y^{a=0}=1, Y^{a=1}=1) + \text{Pr}(Y^{a=0}=1, Y^{a=1}=0)}\]
 
 \[H = \frac{\text{Pr}(Y^{a=0}=0, Y^{a=1}=0)}{\text{Pr}(Y^{a=0}=0, Y^{a=1}=0) + \text{Pr}(Y^{a=0}=0, Y^{a=1}=1)}\]

 Therefore, under our definitions of \textit{G} and \textit{H}, the following relationship holds in any population:  
\[S_1 = S_0 \times G_{\textit{s}} + (1- S_0)\times (1-H_{\textit{s}}).\]

Next, if treatment is monotonically protective, \(H_{\textit{s}} = 1\). The second term is therefore equal to 0, and it follows that \(S_1 = S_0\times G_s\) , and that \(G_s = \frac{S_1}{S_0}= RR(-)_\textit{s}.\)
\end{proof}

\begin{proof}[Proof of Proposition~\ref{proposition2}]
\bigskip
\begin{equation}
\begin{aligned}
T_1&= T_0 \times G_{\textit{t}} + (1- T_0)\times(1-H_{\textit{t}} ) & \text{By same logic as in proposition 1} \\
&= T_0 \times G_{\textit{t}} & \text{By monotonicity} \\
&= T_0 \times G_{\textit{s}} & \text{By equal treatment effects} \\
& \frac{T_1}{T_0} = G_{\textit{s}} = \frac{S_1}{S_0} & \text{By proposition 1} 
\end{aligned}
\end{equation}
\end{proof}

\begin{proof}[Proof of Proposition~\ref{proposition3}]
\bigskip
Define $F$ as \(\frac{T_0}{S_0}\). Then, \(F > 1\) if the baseline risk is higher in the target population, and \(F < 1\) if the baseline risk is lower in the target population. As in the motivating example, our goal is to estimate \(T_1\) from information on  \(T_0\),  \(S_1\) and  \(S_0\). Let \(\hat{T_1}\) be the estimate of \(T_1\). For a protective treatment, if \(\hat{T_1} < T_1\), we can conclude that the prediction is biased away from the null. Our goal is to show that, if the risk ratio is used to transport the effect, and the effects are equal according to Definition 3, then  \( \hat{T_1} < T_1\)  if \(F<1 \). 
If the treatment effects $H$ and $G$ are equal between populations $\textit{s}$ and $\textit{t}$ , we know that the following relationship holds:

\[T_1 = G \times T_0 + (1-H)\times (1-T_0) \]
Alternatively, this can be written in terms of \(S_0\) and $F$:
\[T_1 = G\times S_0\times F + (1-H)\times (1-S_0\times F)\]
The risk ratio $RR(-)$ that will be estimated in population $s$  can be written as 
\[RR(-) = \frac{S_1}{S_0}= \frac{G\times S_0 + (1-H)\times (1-S_0)}{S_0} \]

If we use an approach based on assuming that the risk ratio is equal between the populations, we will estimate \(T_1\) by \( \hat{T_1}\) as follows

\[\hat{T_1} = T_0\times RR(-) = T_0 \times \frac{G \times S_0 + (1-H) \times (1-S_0)}{S_0} = F \times [G\times S_0 + (1-H)\times (1-S_0)]  \]

We will now compare \( \hat{T_1}\) with \(T_1\)  to see which is greater. 
\[T_1 = G \times S_0 \times F + (1-H) \times (1-S_0\times F) = GFS_0 + 1 - H - FS_0 + HFS_0 \]

\[ \hat{T_1} = F\times [G \times S_0 + (1-H) \times (1-S_0)] = GFS_0 + F - FS_0 - FH + HFS_0 \]

Between these two expressions, the terms \(GFS_0\), \(-FS_0\) and \(HFS_0\) and  are shared and cancel. We therefore know that
 
 \[\hat{T_1} < T_1 \text{ if } F-FH  < 1-H  \]
or, equivalently,
 \[\hat{T_1} < T_1 \text{ if }F(1-H)  < 1-H\]

If $H$ is positive, $\hat{T_1} < T_1$  if $F<1$

From the preceding results, we can further derive the bias term \(\hat{T_1} - T_1 = F-FH - 1 + H \), which can be used to graph the amount of bias as a function of the ratio of baseline risks and of the extent of non-monotonicity.
\end{proof}

\begin{proof}[Proof of Proposition~\ref{proposition4}]
\bigskip
As discussed earlier, we know that \(S_1 = S_0 \times G_S + (1-S_0)\times (1-H_{\textit{s}} ) \). If treatment monotonically increases the incidence, \(G_{\textit{s}}  = 1\). We therefore have that \(S_1 = S_0 + (1- S_0) \times (1-H_{\textit{s}})\). Solving this for \(H_{\textit{s}} \), we get
\begin{equation}
\begin{aligned}
H_\textit{s} & = 1- \frac{S_1 - S_0}{1-S_0} \\
  & = \frac{1-S_0 - S_1 + S_0}{1-S_0} \\
  & =  \frac{1-S_1}{1-S_0} \\
  & = RR(+)
\end{aligned}
\end{equation}

\end{proof}

\begin{proof}[Proof of Proposition~\ref{proposition5}]
\bigskip
\begin{equation}
\begin{aligned}
&H_t = \frac{1-T_1}{1-T_0}    &\text {Same logic as in Proposition 4} \\
&H_s	= \frac{1-T_1}{1-T_0} &\text {Equal treatment effects}\\
&H_s	= \frac{1-S_1}{1-S_0} &\text {Proposition 4}\\
&   \frac{1-S_1}{1-S_0}	=  \frac{1-T_1}{1-T_0}\\
\end{aligned}
\end{equation}

\end{proof}

\begin{proof}[Proof of Proposition~\ref{proposition6}]
\bigskip
The logic of the proof of Proposition 6 is exactly analogous to that presented in the context of Proposition 3.

\end{proof}

\section*{Appendix 3}

Consider an unmeasured attribute $X$ that interacts with $A$ to determine treatment response. We will here outline how background biological knowledge can be encoded as restrictions on the joint distribution of counterfactuals of the type $Y^{a,x}$, and show that these restrictions may have implications for effect equality. This will allow us to clarify the link between biology and model choice. For simplicity, we will first consider a situation where response to treatment is fully determined by $X$; and later outline how this assumption can be relaxed. 

For illustration, we will consider an example concerning the effect of treatment with antibiotics ($A$), on mortality ($Y$). We will suppose that response to treatment is fully determined by bacterial susceptibility to that antibiotic ($X$). In the following, we will suppose that attribute $X$ has the same prevalence in populations $s$ and $t$ (for example because the two populations share the same bacterial gene pool) and that treatment with $A$ has no effect in the absence of $X$. Further, suppose that this attribute is independent of the baseline risk of the outcome (for example, old people at high risk of death may have the same strains of the bacteria as young people at low risk). 

In order to get equality of effects between populations $s$ and $t$, we need one further condition: If the attribute $X$ has no effect on $Y$ in the absence of $A$ but prevents $Y$ in the presence of $A$,  the effect of introducing treatment will be equal between the two populations; whereas if $X$ has no effect on $Y$ in the presence of $A$ but causes $Y$ in the absence of $A$, the two populations will have equality of the effect of removing treatment.

The above is formalized with the following conditions:

\begin{enumerate}
\item $X$ is equally distributed in populations $s$ and $t$:  $X \independent P$  
\item  $A$ has no effect in the absence of $X$: $Y^{a=0,x=0}= Y^{a=1,x=0}$ in all individuals
\item 
\begin{enumerate}
\item $X$ has no effect in the absence of $A$: $Y^{a=0,x=0}= Y^{a=0,x=1}$ in all individuals
\item $X$ prevents the outcome in the presence of $A$: $Y^{a=1, x=1} = 0$ in all individuals
\item $X$ is independent of the baseline risk: $X \independent Y^{a=0} \vert P$
\end{enumerate}
\item
\begin{enumerate}
\item $X$ has no effect in the presence of $A$: $Y^{a=1,x=0}= Y^{a=1,x=1}$ in all individuals
\item $X$ causes the outcome in the absence of $A$: $Y^{a=0, x=1} = 1$ in all individuals
\item $X$ is independent of the risk under treatment: $X \independent Y^{a=1} \vert P$
\end{enumerate}
\end{enumerate}

\begin{proposition}
\label{proposition7} If conditions 1, 2 and 3(a-c) hold, then $G$ is equal between populations s and t.  If conditions 1,2 and 4(a-c) hold, then $J$ is equal between the populations. 
\end{proposition}

\begin{proof}[Proof of Proposition~\ref{proposition7}]
\bigskip

\begin{equation}
\begin{aligned}
\small
&G_s =\text{Pr}(Y^{a=1}=1  \vert Y^{a=0}=1, P=s) & & \\
&&\\
&=\text{Pr}(Y^{a=1}=1  \vert Y^{a=0}=1, P=s, X=0)\times \text{Pr}(X=0 \vert  Y^{a=0}=1, P=s)\\
&+ \text{Pr}(Y^{a=1}=1 \vert Y^{a=0}=1, P=s, X=1)\times \text{Pr}(X=1\vert Y^{a=0}=1, P=s)\textit{  (Law of total probability)} \\ 
&&\\
&=\text{Pr}(Y^{a=1, x=0}=1  \vert Y^{a=0, x=0}=1, P=s, X=0) \times \text{Pr}(X=0 \vert  Y^{a=0}=1, P=s)\\
&+ \text{Pr}(Y^{a=1, x=1}=1 \vert Y^{a=0, x=1}=1, P=s, X=1)\times \text{Pr}(X=1\vert Y^{a=0}=1, P=s)\textit{  (Consistency)}  \\ 
&&\\
&=\text{Pr}(Y^{a=1, x=0}=1  \vert Y^{a=0, x=0}=1, P=s, X=0)\times \text{Pr}(X=0 \vert  Y^{a=0}=1, P=s) \\
&+ \text{Pr}(Y^{a=1, x=1}=1 \vert Y^{a=0, x=0}=1, P=s, X=1)\times \text{Pr}(X=1\vert Y^{a=0}=1, P=s)\textit{ (Assumption 3a)}\\
&&\\
&=\text{Pr}(Y^{a=1, x=0}=1  \vert Y^{a=1, x=0}=1, P=s, X=0) \times \text{Pr}(X=0 \vert  Y^{a=0}=1, P=s)\\
&+ \text{Pr}(Y^{a=1, x=1}=1 \vert  Y^{a=1, x=0}=1, P=s, X=1)\times \text{Pr}(X=1\vert Y^{a=0}=1, P=s)\textit{ (Assumption 2)} \\ 
&&\\
& =\text{Pr}(X=0  \vert Y^{a=0}=1, P=s)\textit{ (Assumption 3b)} &\\
&&\\
& =\text{Pr}(X=0  \vert P=s)\textit{( Assumption 3c)}\\
\normalsize
\end{aligned}
\end{equation}
With the same argument, we can show that $G_t=Pr(X=0 \vert P=t)$. By assumption 1, Pr($X=0 \vert P=s)$ and Pr($X=0 \vert P=t$) are equal.
\end{proof}

Results similar to proposition 7 can be shown for attributes $Z$ that are associated with a harmful effect of treatment. This will require conditions for $Z$ that are comparable to 3(a-c), and 4(a-c). We will refer to the conditions that lead to equality of the $H$ parameter as 5(a-c), and the conditions that lead to equality of the $I$ parameter as 6(a-c):

\begin{enumerate}
\setcounter{enumi}{4}
\item
\begin{enumerate}
\item $Z$ has no effect in the absence of $A$: $Y^{a=0,z=0}= Y^{a=0,z=1}$ in all individuals
\item $Z$ causes the outcome in the presence of $A$: $Z^{a=1, z=1} = 1$ in all individuals
\item $Z$ is independent of the baseline risk: $Z \independent Y^{a=0} \vert P$
\end{enumerate}
\item
\begin{enumerate}
\item $Z$ has no effect in the presence of $A$: $Y^{a=1,z=0}= Y^{a=1,z=1}$ in all individuals
\item $Z$ prevents the outcome in the absence of $A$: $Y^{a=0, z=1} = 0$ in all individuals
\item $Z$ is independent of the risk under treatment: $Z \independent Y^{a=1} \vert P$
\end{enumerate}
\end{enumerate}

We now sketch the outline of an argument for how this extends to situations where there is both a protective and a harmful attribute, which together completely determine whether drug $A$ has an effect. Consider the joint counterfactual $Y^{a, x, z}$ where $X$ is an attribute that is associated with a protective effect of treatment with $A$, and $Z$ is an attribute associated with a harmful effect of treatment with $A$, and the joint distribution of $X$ and $Z$ is equal between the populations. It will be necessary that the two attributes are coherent with each other, in the sense that either $X$ meets conditions 3(a-c) and $Z$ meets conditions 5(a-c), or that $X$ meets conditions 4(a-c) and $Z$ meets conditions 6(a-c).

Suppose that for any combination of $X$ and $Z$, treatment with $A$ is either ineffective, protective or harmful. For example, an individual may have a strain of the bacterium that is susceptible to treatment ($X=1)$ , but also have a genetic variant that causes a severe, deadly allergic reaction to the drug ($Z=1$). In this case, we may believe that the allergic reaction supersedes the bacterial susceptibility, i.e. that treatment with $A$ is harmful for this combination of $X$ and $Z$.

For any individual, define $U=0$ if the person belongs to a joint stratum of $X$ and $Z$ such that treatment has no effect (i.e. $Y^{a=1, x, z}$ = $Y^{a=0, x, z}$ in all individuals),  $U=1$ if the individual belongs to a joint stratum of $X$ and $Z$ such that $Y^{a=1, x, z} = 0$ in all individuals, and $U=2$ if he belongs to a joint stratum of $x$,$z$ such that $Y^{a=1, x,z} = 1$. Because no combination of $X$ and $Y$ is associated with both harmful and preventative effects of $A$, this covers all possibilities.

Using logic similar to the proof for singular attributes, we can show that  $G$ is equal to Pr$(U=1)$, $H$ is equal to Pr$(U=2)$, and that these are equal between the two populations. 
 
This can further be extended to multifactorial attributes. Consider the joint counterfactual $Y^{a, \overline{x}, \overline{z}}$ where $\overline{X}$ is a vector of attributes associated with protective effect of treatment, and $\overline{Z}$ is a vector of attributes associated with a harmful effect of treatment. We will suppose that all attributes in $\overline{X}$ either operate according to conditions 3(a-c) or conditions 4(a-c), and that all attributes in $\overline{Z}$ operate according to corresponding conditions 5(a-c) or (6a-c) such that the conditions for $\overline{Z}$ are coherent with the conditions for $\overline{X}, as discussed above.$. This extension will require that for any combination of $\overline{x}$ and $\overline{z}$, treatment with $A$ is either harmful, protective or without effect; if this is believed to be the case, individuals can be assigned to strata of $U$ using the same logic as before.

\section*{Appendix 4}

\begin{proposition}
\label{proposition8} If the outcome is rare and if $RR(+)$ is equal between two populations, then the risk difference is approximately equal between the two populations.
\end{proposition}

\begin{proof}[Proof of Proposition~\ref{proposition8}]
\bigskip
\begin{equation}
\begin{aligned}
&\text{If $RR(+)$ is equal between the two populations, the following relationship holds:}\\
&\\
&\frac{1-\text{Pr}(Y^{a=1}\vert P=s)}{1-\text{Pr}(Y^{a=0}\vert P=s)} =\frac{1-\text{Pr}(Y^{a=1}=1\vert P=t)}{1-\text{Pr}(Y^{a=0}=1\vert P=t)}\\
&\\
&\text{This can be rewritten as:}\\
&\\
&1-\text{Pr}(Y^{a=1}\vert P=s) - \text{Pr}(Y^{a=0}=1\vert P=t) + \text{Pr}(Y^{a=0}=1\vert P=t)\times \text{Pr}(Y^{a=1}\vert P=s) \\
&= 1-\text{Pr}(Y^{a=1}\vert P=t) - \text{Pr}(Y^{a=0}=1\vert P=s) + \text{Pr}(Y^{a=0}=1\vert P=s)\times \text{Pr}(Y^{a=1}\vert P=t)\\
&\\
&\text{If the outcome is rare, the product terms on both sides are close to zero:}\\
&\\
&1-\text{Pr}(Y^{a=1}\vert P=s) - \text{Pr}(Y^{a=0}=1\vert P=t) = 1-\text{Pr}(Y^{a=1}\vert P=t) - \text{Pr}(Y^{a=0}=1\vert P=s)\\
&\\
&\text{This can be rewritten as:}\\
&\\
& \text{Pr}(Y^{a=1}\vert P=s) - \text{Pr}(Y^{a=0}=1\vert P=s) = \text{Pr}(Y^{a=1}\vert P=t) - \text{Pr}(Y^{a=0}=1\vert P=t)\\
&\\
&\text{The left side of this expression is RD in population $s$ and the right side is RD in population $p$:}\\
&\\
& RD_s = RD_t\\
\end{aligned}
\end{equation}
\end{proof}

\bibliographystyle{unsrt}
\bibliography{Mendeley.bib}

\section*{Acknowledgements}
The authors thank James Robins for suggesting the relationship between $RR(+)$ and the Risk Difference,  Etsuji Suzuki for extensive comments on an earlier draft of the manuscript, and Miguel Hernan, Ryan Seals and Steve Goodman for discussions and for insights that improved the manuscript. 

\section*{Author Contributions} AH had the original idea, provided the original version of the theorems and proofs, wrote the first draft of the manuscript and coordinated the research project. AG and SAS contributed original intellectual content and extensively restructured and revised the manuscript. All authors approved the final version of the manuscript.

\section*{Correspondence} All correspondence should be directed to Anders Huitfeldt at The Meta-Research Innovation Center at Stanford, Stanford University School of Medicine, 1070 Arastradero Road, Palo Alto CA 94303; e-mail: anders@huitfeldt.net

Anonymous feedback is welcomed at  http://www.admonymous.com/effectmeasurepaper 

Anders Huitfeldt invokes Crocker's Rules (http://sl4.org/crocker.html) on behalf of all authors for all anonymous and non-anonymous feedback on this manuscript.

\section*{Funding} The authors received no specific funding for this work. While this research was conducted, Dr. Huitfeldt was supported by the Meta-Research Innovation Center at Stanford, which is partly funded by a grant from the Laura and John Arnold Foundation. Dr. Goldstein is supported by National Library of Medicine Training Grant T15LMLM007079. Dr. Swanson is supported by DynaHEALTH grant (European Union H2020-PHC- 2014; 633595).

\end{document}